\documentclass[journal]{IEEEtran}
\usepackage[latin1]{inputenc}
\usepackage{graphicx}
\usepackage{amssymb,amsmath,amsfonts,amsthm}
\usepackage{float}
\usepackage{graphics}
\usepackage{xspace}
\usepackage[usenames,dvipsnames]{color}
\usepackage{amssymb, epsfig, hyperref}
\usepackage{amsmath}
\usepackage{physics}
\usepackage{epstopdf}
\newtheorem{proposition}{Proposition}
\newcommand{\mi}{\mathrm{i}}

\begin{document}

\title{\LARGE Tesla meets Helstrom: a Wireless-Powered Quantum Optical System}

\author{Ioannis Krikidis,~\IEEEmembership{Fellow, IEEE}

\thanks{I. Krikidis is with the Department of Electrical and Computer Engineering, University of Cyprus, Cyprus (e-mail: krikidis@ucy.ac.cy).}}

\maketitle

\begin{abstract}
This letter investigates a novel wireless-powered quantum optical communication system, in which a batteryless quantum transmitter harvests energy from a classical radio-frequency source to transmit quantum coherent states. The transmission employs M-ary phase shift keying (M-PSK) modulation over an optical channel impaired by thermal noise, and the fundamental detection performance is evaluated using the Helstrom bound. An optimization framework is proposed that jointly determines the optimal quantum measurement and the energy-harvesting time fraction to maximize the effective rate under a block time constraint. Analytical expressions are derived for special cases, while semidefinite programming techniques are employed for the general M-PSK scenario. Numerical results validate the unimodal nature of the effective rate function and demonstrate the impact of the optimal design parameters.
\end{abstract}

\vspace{-0.1cm}
\begin{keywords}
Wireless power transfer, RF energy harvesting, quantum optical communications, Helstrom bound, M-PSK. 
\end{keywords}

\vspace{-0.1cm}
\section{Introduction}

\IEEEPARstart{W}{ireless} power transfer (WPT) using dedicated and fully controlled radio-frequency (RF) signals is a key enabler for sixth-generation (6G) wireless communication systems. It supports both energy sustainability and communication connectivity for ultra-low-power and/or batteryless devices, such as Internet-of-Things (IoT) nodes and sensors \cite{KRI}. This technology is primarily driven by recent advancements in electronics and signal processing, which have significantly reduced the power requirements of modern devices (in line with Koomey's law), as well as by notable improvements in RF rectification range and efficiency \cite{BRU}. Additionally, the exponential growth in the number of connected devices renders traditional energy supply methods increasingly impractical, further highlighting the need for scalable wireless energy solutions. 

WPT has been extensively studied in recent years from both a communication and networking standpoint (focusing on waveform design, communication protocols, and network architectures) \cite{KRI,BRU}, as well as from microwave and electronics perspectives, including rectenna circuit design and system architectures \cite{NIO}. More recently, it has also been incorporated into 3GPP studies as a promising enabling technology for passive and ambient IoT systems. A WPT architecture of particular practical interest is the wireless powered communication network (WPCN), where a batteryless transmitter harvests energy from an RF source and uses it immediately for data transmission within a fixed time block \cite{BI}. WPCNs have been studied in various communication scenarios, including multi-user networks, sensor networks, MIMO, semantic communications, etc. \cite{BI, KRI2, KAR}. While most existing works assume RF-based WPT/communication, recent studies explore more sophisticated and flexible setups where hybrid sources harvest energy from RF signals but transmit data over higher frequency bands (and vice-versa), such as optical or visible light communication \cite{LAB}.

On the other hand, quantum communication is an emerging technology that enables the exchange of qubits or classical information by leveraging fundamental principles of quantum mechanics, such as superposition and entanglement \cite{DJO}. This paradigm has garnered significant interest due to its potential to achieve high data rates and inherently secure transmissions. In this work, we focus on the transmission of classical information over a quantum channel, where the information is encoded into quantum states that serve as carriers. In his pioneering work \cite{HEL}, Helstrom introduced the foundations of quantum communication and established a comprehensive mathematical framework for quantum detection and estimation. Among the various quantum communication strategies, coherent-state modulation stands out for its practical implementation and notable energy efficiency, as reliable transmission can be achieved using minimal photon energy \cite{JUN}. These properties create a natural synergy between WPCNs and quantum optical communications, where energy harvested from RF signals can be efficiently used to power low-energy quantum transmitters, paving the way for sustainable and batteryless quantum communication networks.

This letter, for the first time, investigates a WPCN-based quantum optical communication system where a batteryless quantum transmitter harvests energy from a classical RF source to power a quantum link using coherent-state M-ary phase shift keying (M-PSK) modulation. Leveraging the fundamental Helstrom bound, a harvest-then-transmit protocol is designed, and an optimization framework is developed to jointly determine the optimal quantum measurement and the energy harvesting (EH) time fraction that maximize the effective rate under a block time constraint. The framework minimizes the Helstrom error probability for a given EH duration and then performs a one-dimensional search to maximize the effective rate. Closed-form expressions are obtained for BPSK and noiseless M-PSK, while a semidefinite program addresses the general M-PSK case. Numerical results validate the analysis and reveal trade-offs among EH, modulation order, and detection performance. This architecture enables emerging applications such as quantum IoT, low-power quantum sensing, and secure short-range communication, where energy-constrained quantum transmitters are powered via RF signals for quantum optical transmission.

\noindent {\it Notation:} $\Tr(\cdot)$ represents the trace operator, $\mathrm{diag}([x_0,\ldots,x_{M-1}])$ denotes a $M\times M$ diagonal matrix with the specified elements on its main diagonal, $I$ is the identity matrix of appropriate dimension; $X^\dagger$ means the conjugate transpose (also called the Hermitian transpose) of the matrix $X$. A density operator $\rho$ is a Hermitian, positive semidefinite operator with $\Tr(\rho)=1$. An 
ensemble is mixed if any $\rho_i$  has rank $>1$, and pure if all $\rho_i=|\phi_i\rangle \langle \phi_i|$ are rank-one projectors with normalized, possibly non-orthogonal vectors $|\phi_i\rangle$.

\vspace{-0.3cm}
\section{System model}

We consider a WPCN-based quantum optical communication system consisting of three key devices: an RF energy transmitter (ET) that supports the energy sustainability of the optical link, a batteryless quantum transmitter (QT) equipped with an RF EH front-end (rectenna \cite{KRI, BRU}), and a quantum optical receiver (QR). The QT follows a {\it harvest-then-transmit} protocol, whereby it harvests energy from the ET's RF signal to power its laser source and optical modulation components for communication. Specifically, each transmission block (normalized to unit duration) is divided into two phases: i) in the first phase, occupying a fraction $t$ of the block, the QT harvests energy from the received RF signal, ii) in the second phase of duration $(1-t)$, the QT uses the harvested energy to transmit a quantum optical signal to the QR by using M-PSK modulation. RF-based EH is used here for its ubiquity, without loss of generality; the proposed framework can be extended to alternative EH sources, \textit{e.g.}, optical or visible light.

The ET transmits a continuous-wave RF energy signal with fixed power $P$, and the wireless channel between the ET and the QT is assumed to be constant over the considered transmission period, denoted by a deterministic\footnote{This deterministic channel assumption is reasonable due to the typically strong line-of-sight in RF energy harvesting and is sufficient for the scope of this work; block-fading can be considered in future studies.} channel coefficient $h$.
By using a linear EH model \cite{KRI2}, the QT harvests an energy equal to $E=Pt|h|^2$ which is used to generate an electromagnetic field in a coherent state. By assuming a symmetric M-PSK modulation, the coherent state vector is written as \cite[Ch. 5]{DJO}
\begin{align}
|\mu_i\rangle = e^{-\frac{|\mu_i|^2}{2}} \sum_{n=0}^{\infty} \frac{\mu_i^n}{\sqrt{n!}} |n\rangle,
\label{coherent}
\end{align}
where $\mu_i=\sqrt{a}e^{-\frac{\mi 2\pi i}{M}}$ is the complex amplitude of the $i$-th M-PSK symbol with $i=0,\ldots, M-1$, $\mi=\sqrt{-1}$, and $a=E/(1-t)$ is the optical transmit power over the second transmission phase (proportional to the mean photon number per symbol), and $|n\rangle$ denotes the number (Fock) state with $n \in \mathbb{N}$. For numerical tractability, the infinite-dimensional Fock basis in \eqref{coherent} is truncated to a finite size $N_\text{cut}$, preserving the dominant photon number contributions. The optical QT-QR link operates in the presence of thermal background radiation, modelled as a thermal quantum state with noise photons distributed according to the Bose-Einstein statistics. Let $N_a$ denote the average thermal noise photons; then, by using Glauber's P-representation, the density operator of the received field corresponding to the $i$-th M-PSK symbol can be expressed as follows 
\begin{align}
\rho_i=\frac{1}{\pi N_a}\int \exp\left(-\frac{|a-\mu_i|^2}{N_a} \right)|a\rangle \langle a|  d^2 a. \label{noise}
\end{align}
By using the generalized Laguerre polynomials, the elements  $\rho_i(m,n)$ of the $\rho_i$ matrix are given in closed form by \cite[eq. 5.135]{DJO}

\vspace{-0.1cm}
\begin{align}
\langle m | \rho_i | n \rangle =
\begin{cases}
\displaystyle
\frac{e^{-|\mu_i|^2 / (N_a + 1)}}{N_a + 1}
\sqrt{\frac{n!}{m!}}
\left( \frac{\mu_i^*}{N_a} \right)^{m - n}
\left( \frac{N_a}{N_a + 1} \right)^n \\
\hspace{1.5cm}
\times \, L_n^{(m-n)}\left( -\frac{|\mu_i|^2}{N_a(N_a + 1)} \right),
\quad m \ge n, \\[6pt]
\displaystyle
\langle n | \rho_i | m \rangle^*, \quad m < n,
\end{cases} \label{rho_elements}
\end{align}
where $L_n^{(\nu)}(x)$ denotes the associated Laguerre polynomial of degree $n$ and order $\nu$.

The quality of the quantum optical link is characterized by the symbol error probability, {\it i.e.}, the probability that the receiver selects an M-PSK symbol different from the one transmitted. In this work, we analyze the fundamental performance limit by evaluating the Helstrom bound \cite{HEL}, rather than considering a specific receiver architecture. The Helstrom bound quantifies the minimum achievable error probability under optimal quantum measurements and depends on the quantum states associated with each symbol and the corresponding measurement operators {\it i.e.,} set of positive operator-valued measure (POVM) matrices $\Pi_i$ (a generalization of projective measurements that enables optimal discrimination among non-orthogonal states \cite[Ch. 3.3.1]{DJO}). While such optimal receivers may be difficult to realize in practice, the Helstrom bound provides a benchmark against which more practical detection strategies ({\it e.g.,} Kennedy, Dolinar etc.) can be evaluated. We introduce a novel optimization problem that jointly selects the detection POVM operators $\Pi_i$ and the optimal transmission time fraction $t$ in order to maximize the effective rate, defined as the average number of bits successfully transmitted per channel use. Specifically, we have 
\begin{align}
[P1]\;\;\;&\max_{\{ t,\Pi_i \}}R(t,\Pi_i)=\log_2(M)(1-t)(1-P_e(t,\Pi_i)), \label{eq1}\\
&\quad \text{s.t. } \sum_{i=1}^M \Pi_i = I, \quad \Pi_i \succeq 0, \label{eq2} \\
&\;\;\;\;\;\;\;\;\;\;0<t<1, \label{eq3}
\vspace{-0.3cm}
\end{align} 
where $P_e$ denotes the average error probability (Helstrom bound). The objective function in \eqref{eq1} is the effective rate, measured in bits per channel use (BPCU), which serves as a key metric to capture the trade-off between EH and data transmission time. The constraints in \eqref{eq2} enforce the completeness and positive-semidefinite properties of the POVM measurement matrices, while \eqref{eq3} reflects the two-phase structure of the proposed protocol. Equal priors are assumed for all M-PSK symbols.

\vspace{-0.3cm}
\section{Optimal design and special cases}

In this section, we examine the solution to the above problem and present special cases that admit analytical expressions for the Helstrom bound.

\vspace{-0.3cm}
\subsection{Noiseless quantum channel with BPSK} 

In the first special case of interest, the transmitted symbols $0$ and $1$ correspond to the pure coherent states $|\mu\rangle$ and $|-\mu\rangle$, respectively. The quantum optical channel is assumed to be noiseless, with $N_a = 0$. Under these conditions, the received density operators are given by $\rho_0 = |\mu\rangle\langle\mu|$ and $\rho_1=|-\mu\rangle\langle -\mu|$. The POVM matrices are defined as $\Pi_0 = |\eta_0\rangle\langle \eta_0|$ and $\Pi_1 = I - \Pi_0$, where $|\eta_j\rangle$ and $\eta_j$ denote the $j$-th eigenket and eigenvalue, respectively, of the Hermitian operator $\Delta=\rho_1-\rho_0$, given by the eigenvalue equation \cite{HEL}
\begin{align}
\Delta|\eta_j\rangle = \eta_j|\eta_j\rangle. \label{spectrum}
\end{align}
It is worth noting that, since the density matrices $\rho_0$ and $\rho_1$ correspond to pure states ({\it i.e.,} rank-one operators), the Hermitian operator $\Delta$ has rank two and admits exactly one positive and one negative eigenvalue (symmetric around zero for the case of symmetric prior probabilities). This structure arises from the binary quantum hypothesis testing framework for pure states \cite[Ch. 5.3]{DJO}. In this case, the Helstrom bound admits an analytical expression given by
\begin{align}
P_e(t)&=\frac{1}{2\!}\left(1-\sqrt{1-|\langle \mu|\!-\!\mu\rangle|^2} \right)=\frac{1}{2}\left(1\!-\!\sqrt{1-e^{-4|\mu|^2}} \right) \nonumber \\
&=\frac{1}{2}\left(1-\sqrt{1-e^{-\frac{4P|h|^2 t}{1-t}}} \right).
\end{align}
The effective rate $R(t)=(1-t)(1-P_e(t))$ is unimodal, as shown by analyzing the properties of its first derivative {\it i.e.,} $R'(t)$ is strictly decreasing on $(0,1)$ and has opposite signs at the two ends of the interval; hence, it crosses zero exactly once, confirming its unimodal structure. The optimal $t^*$ is obtained by numerically solving the stationarity condition, namely the value of $t$ for which $R'(t)=0$.
\begin{proposition}\label{prop1}
In both the high-power regime ($P \rightarrow \infty$) and the low-power regime ($P \rightarrow 0$), the optimal time fraction allocated to EH vanishes, i.e., $t^* \rightarrow 0$.
\end{proposition}
\begin{proof}
We observe that $P_e(t) \to 0$ as $P \to \infty$ for any $t > 0$, hence $R(t) \to 1-t$, which is maximized at $t^* \rightarrow 0$ (sufficient energy is harvested even for vanishingly small $t$). Similarly, as $P \to 0$, we have $P_e(t) \to \frac{1}{2}$ (random guess), so $R(t) \to \frac{1-t}{2}$, also maximized at $t^*  \rightarrow 0$.
\end{proof}
The above behaviour is explained by the exponential nature of the Helstrom bound; at high power, even a small harvesting time $t$ yields sufficient energy for near-perfect detection, making it optimal to maximize communication time. At low power, detection remains unreliable regardless of harvesting time, so again it is preferable to prioritize communication. In both cases, performance is maximized as $t^* \to 0$.

\vspace{-0.2cm}
\subsection{Noisy quantum channel with BPSK} 

In this special case, we consider BPSK modulation over a thermal (noisy) quantum optical channel, consistent with the general system model. Although the Helstrom bound does not admit a simple closed-form expression as in the noiseless case, it still allows for a near-analytical evaluation. Specifically, the received density operators $\rho_0$ and $\rho_1$ correspond to displaced thermal states with average noise photon number $N_a > 0$, and are computed using the expression in \eqref{rho_elements} with $\mu_0 = \mu$ and $\mu_1 = -\mu$, respectively. The minimum error probability is then evaluated numerically as \cite{HEL}
\begin{align}
P_e(t) = \frac{1}{2}\left(1 - \sum_{\eta_i > 0} \eta_i \right),
\end{align}
by considering only the positive eigenvalues $\eta_i$ of the Hermitian operator $\Delta$. To compute the eigenvalues of $\Delta$, we first express the density operators $\rho_0$ and $\rho_1$ in a truncated Fock space of dimension $N_{\text{cut}}$, yielding finite-dimensional matrix representations. The Hermitian operator $\Delta$ is then diagonalized numerically via eigenvalue decomposition (see \eqref{spectrum}), resulting in $N_{\text{cut}}$ real eigenvalues. Due to the symmetry in prior probabilities, $\Delta$ is traceless $ (\Tr(\Delta) = 0)$, and thus its eigenvalues are symmetric around zero. This procedure yields a numerically accurate estimation of the Helstrom bound for binary coherent-state discrimination in the presence of thermal noise. We also note that the optimal POVM in this case is defined by projection onto the positive and non-positive eigenspaces of $\Delta$, respectively {\it i.e.,}  $\Pi_0 = \sum_{\eta_i > 0} |\eta_i\rangle \langle \eta_i|$, $\Pi_1 = I - \Pi_0$. Since a near-closed-form expression for $P_e$ as a function of time is available, and given the unimodality of the effective rate function, the solution to the corresponding optimization problem can be efficiently obtained via numerical search over a single variable. In practice, the time interval is symmetrically discretized into $K$ points within the open interval $(0,1)$, and an exhaustive search is performed by computing and comparing the candidate values $R(t_i)$, where $i=1,\dots,K$.

\vspace{-0.5cm}
\subsection{Noisy quantum channel with M-PSK and $M>2$}

For the general noisy M-PSK case, the Helstrom bound does not admit closed-form or near-analytical expressions, and solving the original problem [P1] requires a more advanced optimization framework. Specifically, the Helstrom bound can be formulated as a semidefinite program (SDP) over a set of POVM matrices $\{\Pi_i\}$ that minimize the average error probability. For a given EH fraction $t$, the Helstrom bound can be written as
\begin{align}
[P2]\;\;\;&\min_{\{\Pi_i \}} P_e(t)=\frac{1}{M}\sum_{i=1}^M \Tr\left(  (I - \Pi_i)\rho_i \right), \label{eq11}\\
&\quad \text{s.t. } \sum_{i=1}^M \Pi_i = I, \quad \Pi_i \succeq 0,
\end{align}
where the objective function in \eqref{eq11} represents the Helstrom bound $P_e(t)=1-\frac{1}{M}\sum_i\Tr(\Pi_i\rho_i)=\frac{1}{M}\sum_i\Tr((I-\Pi_i) \rho_i)$. As before, to make the problem tractable, the infinite-dimensional Fock space is truncated to a finite dimension $N_\text{cut}$. The resulting finite-dimensional problem is a standard SDP and can be efficiently solved using convex optimization tools such as \texttt{CVX}\footnote{The SDP is solved using interior-point methods, with polynomial per-iteration complexity $\mathcal{O}(M^2 N_{\text{cut}}^6)$. The method is tractable for moderate $M$; scalability is polynomial but becomes costly for large constellations.}. Since the Helstrom bound is computed numerically through this framework, and given the unimodality of the effective rate function, the solution to [P1] can be efficiently obtained by applying the same single-variable numerical search approach described earlier. Specifically, the time interval is discretized, and the corresponding effective rate values are evaluated as $R(t_i) = \log_2(M) (1 - t_i)(1 - P_e(t_i))$ for each discretization point $t_i$.

\vspace{-0.3cm}
\subsection{Noiseless quantum channel with M-PSK and $M>2$}
For the case of a noiseless quantum channel with M-PSK modulation and symmetric prior probabilities (which is the case considered in this work), the POVM detection operators admit a closed-form expression given by the square-root measurement (SRM) \cite[Ch. 5.10]{DJO}. In this case, there is no need to solve an SDP, as the optimal POVM matrices are explicitly given by
\begin{align}
\Pi_i = \rho^{-1/2}\rho_i \rho^{-1/2} =\rho^{-1/2} |\mu_i\rangle\langle \mu_i| \rho^{-1/2},
\end{align}
where $\rho$ is the unormalized (ensemble) state defined as $\rho=\sum_{i=1}^M \rho_i$. The corresponding Helstrom bound for a given EH time fraction $t$ can then be computed by using the error probability definition in \eqref{eq11} as $P_{e}(t)=1-(1/M)\sum_i\Tr(\rho^{-\frac{1}{2}}\rho_i \rho^{-\frac{1}{2}}\rho_i)=1-(1/M)\sum_i \left( \langle\mu_i|\rho^{-\frac{1}{2}}|\mu_i\rangle \right)^2$. A closed-form expression for $P_e(t)$ is derived in Appendix \ref{ap1}. For large/finite $M$, we derive a simple closed-form expression for the SRM error probability, as stated in the following proposition 
\begin{proposition}
For a large and finite modulation order $M\gg \alpha$, the SRM error probability is approximated as 
\begin{align}
	P_e(t)=1-\frac{e^{-\frac{P|h|^2 t}{1-t}}}{M}\left(\sum_{k=0}^{M-1} \sqrt{ \frac{\left(\frac{P|h|^2t}{1-t}\right)^k}{k!}} \right)^2 \label{app2e}
\end{align}
\end{proposition}
\vspace{-0.5cm}
\begin{proof}
See Appendix \ref{ap2}.
\end{proof}
Given the closed-form expressions of $P_e(t)$, the solution to [P1] can be obtained numerically by employing a single-variable optimization framework, similar to the one described earlier. For the special case of large $M$, we also observe that the statement of Proposition \ref{prop1} also holds. For the low $P$ regime, the proof is provided in Appendix~\ref{ap3}, whereas for the high $P$ regime, it is validated numerically since the above expression holds for $M \gg \alpha$ and thus cannot be applied as $P \to \infty$.

\begin{figure}[t]
\centering
\includegraphics[width=0.84\linewidth]{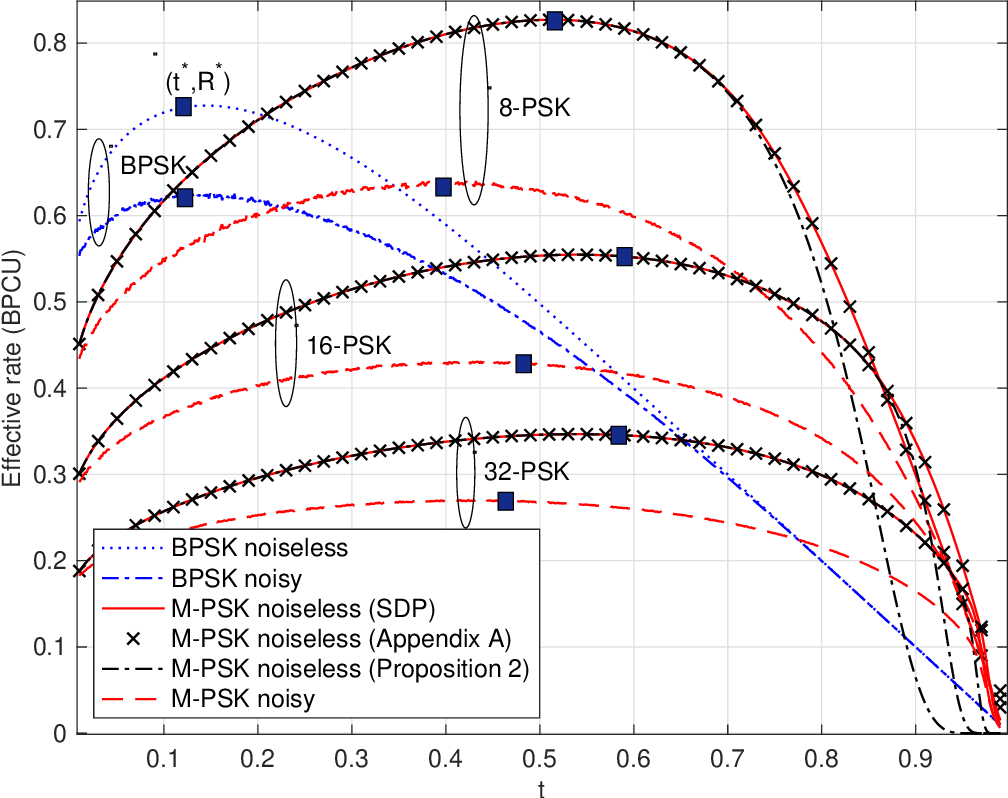}
\vspace{-0.1cm}
\caption{Effective rate versus EH fraction $t$; $P=1$, $M=\{2, 8, 16, 32 \}$, and $N_a=\{0, 0.5 \}$.}\label{fig1}
\end{figure}

\begin{figure}[t]
\centering
\includegraphics[width=0.84\linewidth]{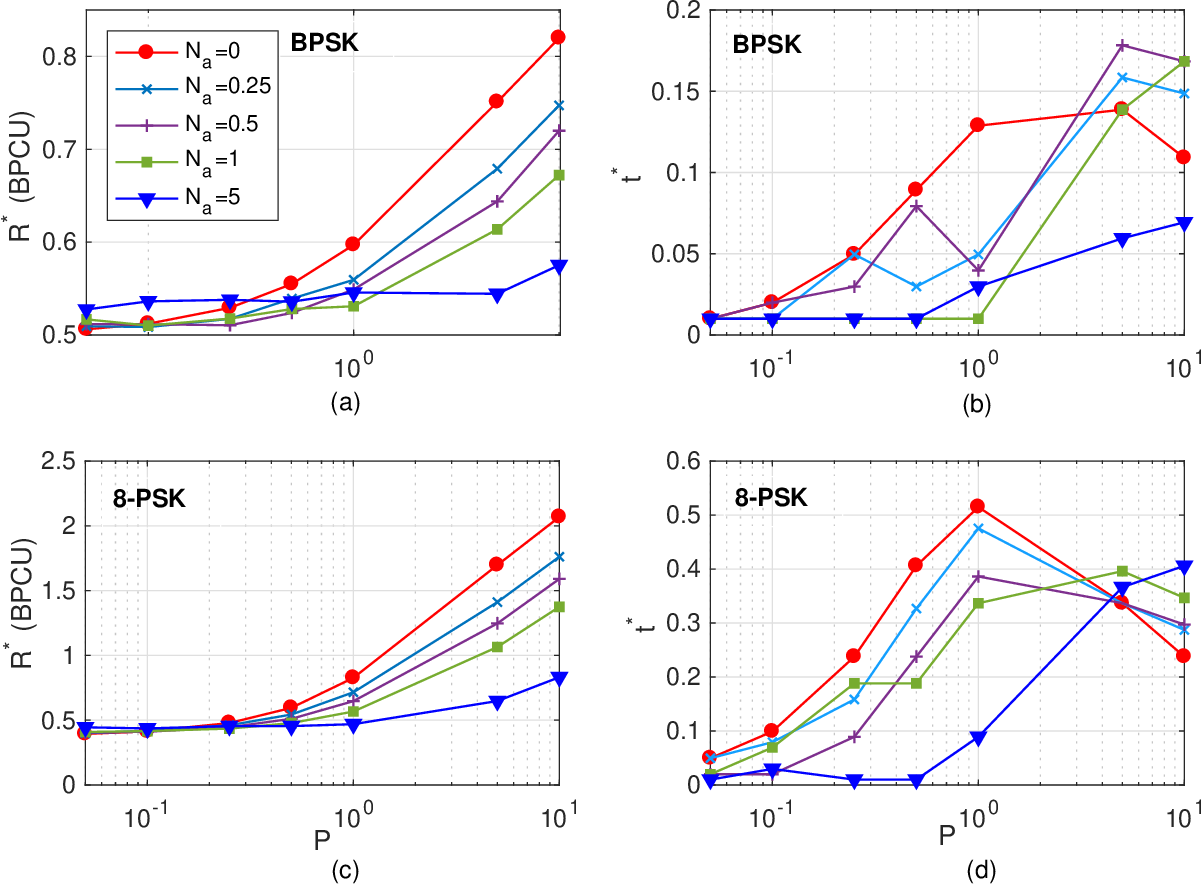}
\vspace{-0.25cm}
\caption{(a), (c) Maximum effective rate ($R^*$) versus transmit power $P$ for BPSK, and $8$-PSK, respectively. (b), (d) Optimal time fraction ($t^*$) versus transmit power $P$ for BPSK, and $8$-PSK, respectively; $N_a=\{0, 0.25, 0.5, 1, 5 \}$.}\label{fig2}
\end{figure}

\vspace{-0.4cm}
\section{Numerical results \& discussion}

The performance of the proposed WPCN-based design is validated through numerical evaluations; the simulation setting assumes $|h|^2=1$, $N_\text{cut}=40$, and $K=1000$ (time resolution).

Fig.~\ref{fig1} illustrates the effective rate as a function of the time fraction $t$ for a system configuration with $P = 1$, $N_a = \{0, 0.5\}$, and $M = \{2, 8, 16, 32\}$. The effective rate demonstrates a distinct unimodal behavior, featuring a single peak. As the modulation order increases, the maximum achievable effective rate decreases, while the optimal time fraction $t^*$ shifts to larger values. This trend reflects the fact that higher-order PSK constellations demand more energy to ensure reliable symbol detection at the receiver. Moreover, the presence of Gaussian noise further reduces the effective rate due to its detrimental effect on detection performance. Interestingly, the impact of noise diminishes with increasing modulation order, as the symbol spacing becomes the dominant factor in the error probability. For $M > 2$, the results also confirm the accuracy of the closed-form expression provided in Appendix~\ref{ap1}, along with the approximation in~\eqref{app2e}, which becomes increasingly precise as $M$ grows and satisfies $M \gg \alpha$. The optimal $t^*$ from the noiseless case serves as a useful approximation under noise. As thermal noise increases, reduced quantum state distinguishability shifts the EH-communication trade-off, widening the gap; for BPSK, this gap remains small due to the stronger separation between quantum states.

Figures ~\ref{fig2}(a) and \ref{fig2}(c) present the maximum effective rate as a function of the transmit power~$P$ for various values of thermal noise~$N_a$, for BPSK and $8$-PSK modulation schemes, respectively. As expected, the system performance improves with increasing transmit power, since the quantum transmitter harvests more energy, thereby enhancing the detection probability. This improvement is more pronounced for $8$-PSK due to its higher sensitivity to energy availability. Conversely, thermal noise generally degrades the effective rate; as the average noise level~$N_a$ increases, performance deteriorates, with $8$-PSK being more susceptible to noise-induced degradation than BPSK. An interesting observation, however, is that in the very low-power regime, the effective rate slightly increases with~$N_a$. This counterintuitive behavior arises because, for extremely weak coherent states, added thermal noise slightly reduces the overlap between symbols, thereby improving their distinguishability under the Helstrom bound. 

Figures ~\ref{fig2}(b) and \ref{fig2}(d) show the corresponding optimal time fraction~$t^*$ that maximizes the effective rate for BPSK and $8$-PSK, respectively. For low~$N_a$ values, the optimal fraction exhibits a concave trend {\it i.e.,} it initially increases with power up to a certain point and then decreases as power continues to grow (with $t^*\rightarrow 0$ as $P\rightarrow \infty$). For higher~$N_a$ levels, the optimal~$t^*$ increases monotonically with power. Additionally, it can be seen that for a given transmit power, the optimal harvesting time fraction decreases as the noise level increases.

Overall, this work presents a novel optimization framework that jointly designs quantum measurement and EH time to maximize the effective rate in a WPCN-assisted quantum optical link. Theoretical and simulation results reveal a fundamental trade-off between EH and transmission, confirming its effectiveness for BPSK and M-PSK in both noisy and noiseless cases. Future work includes its extension to scenarios with channel fading, non-linear energy harvesting models, and practical quantum receivers (e.g., Dolinar).

\vspace{-0.1cm}
\appendices

\vspace{-0.1cm}
\section{SRM error probability}\label{ap1}
To provide an analytical expression for the error probability, firstly we compute the Gram matrix with elements \cite{HEL}
\vspace{-0.1cm}
\begin{align}
G(i,j)=\langle \mu_i | \mu_j \rangle = e^{-\alpha + \alpha e^{2\pi \mi (j - i)/M}}.
\end{align}
The corresponding Gram matrix $G \in \mathbb{C}^{M \times M}$ is a circulant matrix\footnote{A matrix \( G \in \mathbb{C}^{M \times M} \) is called \emph{circulant} if there exists a vector $[\gamma_0, \gamma_1, \dots, \gamma_{M-1}] \in \mathbb{C}^M$ such that $G(i,j) = \gamma_{(j - i) \bmod M},\;\;\text{for all}\; i, j = 0, \dots, M-1$.} with first row
\begin{align}
\gamma_\ell = \langle \mu_0 | \mu_\ell \rangle\;=f_{\alpha}\left(\frac{\ell}{M} \right),\;\textrm{with}\;\ell=0,\ldots, M-1,
\end{align}
where $f_\alpha(x) =e^{-\alpha [1 - \cos(2\pi x)] + \mi \alpha \sin(2\pi x)}$ is $1$-periodic function. Let $\lambda_k$ be the eigenvalues of $G$, given by the discrete Fourier transform (DFT) of the first row; the eigenvalues are guaranteed to be real due to conjugate symmetry in $\gamma_\ell$ ({\it i.e.,} $\gamma_{M - \ell} = \overline{\gamma_\ell}$) and the structure of the DFT, so we have
\begin{align}
\lambda_k &= \sum_{\ell=0}^{M-1} f_{\alpha}\left(\frac{\ell}{M} \right)e^{-\frac{2\pi \mi k \ell}{M}},\;\;\textrm{with}\;k = 0, \dots, M-1.\label{exp1}
\end{align}
Since $G$ is circulant, it can be diagonalized by the DFT matrix $F$ (with $F(i,j)=\frac{1}{\sqrt{M}}e^{-\frac{2\pi\mi (i\times j)}{M}}$ where $i,j=0,\ldots,M-1$) as $G=F^\dagger \mathrm{diag}([\lambda_0,\ldots,\lambda_{M-1}])F$, and the SRM operator is defined as \cite[Ch. 5.10.1]{DJO}
\begin{align}
G^{-1/2} G &= F^\dagger \mathrm{diag}([\sqrt{\lambda_0},\ldots,\sqrt{\lambda_{M-1}}])F.
\end{align}
By using simple algebraic computations, we can show that the diagonal elements of the above matrix are equal, therefore
\begin{align}
 \left[G^{-1/2} G\right]_{ii}=\frac{1}{M}\sum_{k=0}^{M-1}\sqrt{\lambda_k}\;\;\textrm{for}\;\;i=1,\ldots,M-1.
\end{align}
The SRM error probability can be expressed equivalently as a function of the Gram matrix as  
\begin{align}
&P_e(t)=1-\frac{1}{M}\sum_i \left( \langle\mu_i|\rho^{-\frac{1}{2}}|\mu_i\rangle \right)^2 \nonumber \\
&\!=\! 1\!-\!\frac{1}{M} \sum_{i=0}^{M-1} \left| \left[G^{-1/2} G\right]_{ii} \right|^2\!=\!1\!-\!\left( \frac{1}{M} \sum_{k=0}^{M-1} \sqrt{\lambda_k} \right)^2. \label{pe}
\end{align}
\vspace{-0.2cm}
\section{SRM error probability for large and finite $M$}\label{ap2}

Extending the results of Appendix A, asymptotically for large $M$, the eigenvalue in \eqref{exp1} admits a closed-form Fourier series representation. Specifically, we have  
\begin{align}
\frac{1}{M}\sum_{\ell=0}^{M-1} f_\alpha\left( \frac{\ell}{M} \right) e^{-2\pi \mi k \ell / M}\rightarrow \int_0^1 f_\alpha(x) \, e^{-2\pi \mi k x}dx.\label{e2}
\end{align}
The above expression requires a sampling resolution $1/M$ fine enough to resolve the variations of $f_\alpha(x)$, which leads to the condition $M \gg \alpha$. By using Taylor expansion $f_\alpha(x)= e^{-\alpha} \sum_{n=0}^\infty \frac{\alpha^n}{n!} e^{2\pi \mi n x}$ and \eqref{e2} can be written as
\begin{align}
\int_0^1 f_\alpha(x) \, e^{-2\pi \mi k x}dx = e^{-\alpha}\frac{\alpha^k}{k!}. \label{e3}
\end{align}
By combining  \eqref{e3}, \eqref{e2}, and substituting into \eqref{exp1}, we have
\begin{align}
&\lambda_k \approx M e^{-\alpha}\frac{\alpha^k}{k!}\Rightarrow
\frac{1}{M} \sum_{k=0}^{M-1} \sqrt{\lambda_k} \approx \frac{1}{\sqrt{M}}\sum_{k=0}^{M-1} \sqrt{ e^{-\alpha}\frac{\alpha^k}{k!}},
\end{align}
and thus the SRM error probability in \eqref{pe} becomes
\begin{align}
P_e(t)=1-\frac{e^{-\alpha}}{M}\left(\sum_{k=0}^{M-1} \sqrt{ \frac{\alpha^k}{k!}} \right)^2\;\text{with}\;M\gg \alpha. \label{exx1}
\end{align}
\section{SRM error probability for large $M$ and low $P$}\label{ap3}

For $P \rightarrow 0$, we have $\alpha \rightarrow 0$, and thus equation~\eqref{exx1} simplifies to $P_e(t) = 1 - \frac{1}{M}$. In this regime, the effective rate becomes $R(t) = \frac{\log_2(M)(1 - t)}{M}$, which attains its maximum value $R^* \rightarrow \frac{\log_2(M)}{M}$ as $t^* \rightarrow 0$.

\end{document}